\documentclass[12pt]{article}
\usepackage{titlesec}

\usepackage{xcolor}

\usepackage{mathtools, latexsym, amsmath, amssymb,amsthm, natbib,marvosym, verbatim,setspace}
\usepackage{amssymb}

\usepackage[utf8]{inputenc}

\usepackage[a-1b]{pdfx}
\hypersetup{colorlinks=true, pdfstartview=FitV, linkcolor=blue,citecolor=blue, urlcolor=blue}
\usepackage{geometry}
\usepackage{appendix}

\def\Ob{\mathcal{O}}

\def\Re{\mathbf{R}}

\def\al{\alpha}

\def\phi{\varphi}

\def\os{\emptyset}

\newcommand{\df}[1]{\textbf{#1}}

\titleformat{\paragraph}
  {\normalfont\normalsize\bfseries}
  {}                   
  {1em}                             
  {}
\titlespacing*{\paragraph}
  {0pt}
  {3.25ex plus 1ex minus .2ex}
  {1.5ex plus .2ex}

\newtheorem{proposition}{Proposition}
\newtheorem{lemma}{Lemma}

\newtheorem{definition}{Definition}
\newtheorem{remark}{Remark}
\newtheorem{example}{Example}
\newtheorem{claim}{Claim}

\setcitestyle{square,authoryear}
\def\cite{\citet}

\title{Swap Bounded Envy}
\author{Federico Echenique\thanks{Department of Economics, UC Berkeley, {fede@econ.berkeley.edu}}, Sumit Goel\thanks{Division of Social Sciences, New York University, Abu Dhabi, {sumitgoel58@gmail.com}}, and SangMok Lee\thanks{Department of Economics, Washington University in St. Louis, {sangmoklee@wustl.edu}}}
\date{\today}

\begin{document}
\maketitle

\begin{abstract}
We study fairness in the allocation of discrete goods. Exactly fair (envy-free) allocations are impossible, so we discuss notions of approximate fairness. In particular, we focus on allocations in which the swap of two items serves to eliminate any envy, either for the allocated bundles or with respect to a reference bundle. We propose an algorithm that, under some restrictions on agents' preferences, achieves an allocation with ``swap bounded envy.''

\end{abstract}

\begin{quote}

\setcounter{page}{0}\thispagestyle{empty}
\end{quote}

\clearpage

\section{Introduction}

We discuss the meaning and feasibility of \df{fair allocations} of bundles of indivisible items to economic agents. Consider the following situation. Alice and Bob are students at a four-year college, and each is assigned a dorm room for their freshman, sophomore, junior, and senior years. Let Alice's bundle of rooms be $(o^A_1, o^A_2, o^A_3, o^A_4)$ and Bob's be $(o^B_1, o^B_2, o^B_3, o^B_4)$. The assignment is made by a central university office. When can we say that this assignment is fair?

An ideal situation occurs when Alice and Bob's preferences are distinct enough that they are each happy with the bundle of items they have received. In this case, Alice does not envy Bob, and Bob does not envy Alice. Outside of such ideal situations, the assignment of indivisible items must entail some unfairness. For example, if Alice and Bob want the exact same rooms, any assignment will leave one of them envying the other. We, therefore, focus on situations where the amount of envy can be limited or bounded. In particular, we discuss three notions of approximate fairness, one of which is new to our paper.

To fix ideas, imagine that Alice envies Bob. She regards Bob's bundle as superior to her own. If Alice complains, the central office might respond with an argument, not an action: ``Yes, you are correct, but our mechanism ensures that your envy is due to only \emph{one} of Bob's assigned rooms. If that single room were removed from Bob's bundle, your envy would be eliminated.'' This response is purely an assessment of the degree of Alice's envy. The office is not proposing to actually change Bob's assignment. The underlying mechanism bounds the envy, in the sense that the elimination of only one item from the envied bundle would remove Alice's envy. This property is well-known in the literature \citep{lipton2004approximately,budish2011combinatorial} as \df{envy-free up to one good (EF1)}.

For dorm assignments, EF1 has a clear problem. Imagine that eliminating Alice's envy requires removing Bob's senior year dorm assignment. Where is Bob supposed to live? This hypothetical removal implies an outside option---such as off-campus housing---that Bob can use. The notion of EF1, therefore, rests on the agents' valuation of such an outside option. If the college has terrible off-campus options, EF1 becomes very easy to satisfy; removing any of Bob's dorm rooms would make his resulting bundle so unattractive that Alice's envy is immediately eliminated. This highlights that EF1 can sometimes be too weak a fairness requirement.

The possible weakness of EF1 motivates looking for a different bound on envy. Suppose the campus office can tell Alice: ``True, you envy Bob. But this envy is limited because there is one room in your bundle that can be swapped with one room in his, and this single swap would remove your envy.'' Again, this is not a practical proposal to swap rooms but an argument to demonstrate that the envy tolerated by the system is bounded. The envy can be traced to a single pair of swappable items. The resulting property is called \textbf{swap envy-free (swapEF)} and has been studied by \cite{caragiannis2024repeatedly}, \cite{bogomolnaia2025teams} and \cite{mancho2025fairness}. 

Swap envy-freeness has one undesirable property. Imagine that Alice envies Bob and that a swap suffices to remove the envy. Now there is a reassignment of dorm rooms that leaves Alice better off without changing Bob's assignment. Despite the reassignment, Alice may still envy Bob. Moreover, it is possible that this envy cannot be removed by a swap. So the property of being swapEF is not preserved by a reallocation that improves the bundle of Alice, the envying partner.\footnote{See Remark~\ref{rmk:TTC_swapEF} for a specific example.} 

Thus motivated by the disadvantages of EF1 and swapEF, 
we introduce a notion of bounded envy that is robust to a welfare improvement for the envying agent. The office's response to Alice could involve a hypothetical reference bundle, $(o_1,o_2,o_3,o_4)$, and a single year, $t$. Say $t=4$, the senior year, just to be concrete. This reference bundle serves to bound Alice's envy if two conditions hold:
\begin{enumerate}
    \item Alice considers her own bundle to be at least as good as the reference bundle.
    \item Alice would prefer to swap the year-$t$ room from the reference bundle for the year-$t$ room from Bob's bundle. That is, she would rank the bundle $(o_1,o_2,o_3,o^B_4)$ as 
    better than $(o^B_1,o^B_2,o^B_3,o_4)$.
\end{enumerate}

This argument demonstrates that Alice's envy for Bob's entire bundle is largely concentrated in a single item ($o^B_4$). While she might envy Bob's complete package, his whole four-year trajectory in campus housing, the value is not distributed evenly. One specific room is the primary source of envy. This property is new to our paper. We call it \df{swap bounded envy}, or swapBE.

We focus on general problems where agents receive bundles of the same cardinality. We show that the general class of \df{draft mechanisms}, where agents choose items in sequence, produces assignments that are both EF1 and swapEF. 

Next, we consider the general version of the dorm assignment problem, in which agents must consume a bundle of a given cardinality, and an outside option exists for bundles that fall short. Here, a draft mechanism may not achieve EF1 or swapEF. We propose a mechanism based on Top Trading Cycles and Serial Dictatorship, and show that under a preference restriction (separability), it produces an assignment that is EF1 and swapBE.

Unfortunately, the above mechanisms may result in inefficient allocations. We then turn to a Nash welfare maximization objective and show that, under certain strong restrictions on agents' preferences, maximizing Nash welfare results in an assignment that is EF1, swapEF, and Pareto optimal among equal-cardinality assignments.

\paragraph{Related Literature.}

\cite{caragiannis2024repeatedly} study repeated assignment of a set of indivisible items to a set of agents; the dorm assignment problem is a special case of their model. Their main results concern social welfare maximization and the existence of EF1 allocations, but they introduce the notion of swapEF that we have discussed in the introduction.

\cite*{bogomolnaia2025teams} and \cite*{mancho2025fairness} study allocation problems in which the cardinality of bundles is fixed, and consider swapEF (termed envy-free up to one flip in their papers). Under additive utilities, they show that swapEF  allocations always exist. \cite{bogomolnaia2025teams} also show compatibility with efficiency in some special cases, such as two agents, identical utilities, and binary utilities, while \cite{mancho2025fairness} prove that a Nash welfare–maximizing allocation is approximately swapEF.

We are motivated by problems like the one we highlighted in the introduction, in which EF1 is desirable but may be too weak to be a convincing bound in envy. 
We discuss algorithms that achieve swapEF or the new notion that we introduce, swap-bounded envy, along with EF1.

\cite*{unreasonable_fairness} shows that Nash welfare maximization finds an EF1 and Pareto-efficient allocation. Their results do not hold when there are restrictions on the bundles that may be received, such as the four-year assignment of dorm rooms to students (or, more generally, what we call multi-dimensional assignment). \cite*{biswas2018fair} consider allocation problems under a cardinality constraint, such as the dorm room assignment, and prove the existence of EF1 allocations for additively separable preferences. \cite{unreasonable_fairness} introduce the notion of being envy-free up to any good, meaning that Alice's envy of Bob is eliminated by the removal of any of Bob's items in his bundle. \cite{plaut2020almost} provide conditions under which allocations that are envy-free up to any good exist.

\section{Setup}

A general allocation problem consists of a tuple $(\mathcal{N}, \mathcal{O}, (\succeq_i)_{i \in \mathcal{N}})$, where
\begin{itemize}
    \item $\mathcal{N} = \{1, 2, \dots, n\}$ is a set of agents,
    \item $\mathcal{O} = \{o_1, \dots, o_m\}$ is a set of indivisible objects, and
    \item $\succeq_i$ is each agent $i$'s preference over $2^{\mathcal{O}}$.
\end{itemize}
An allocation is a function $\mu: \mathcal{N} \to 2^{\mathcal{O}}$ such that $\mu(i) \cap \mu(j) = \emptyset$ for $i \neq j$, 
and $\cup_{i} \mu(i) = \mathcal{O}$.\footnote{The last requirement rules out a trivially envy-free allocation where no agent receives any object.}

\paragraph{Properties of preferences.}

A preference $\succeq$ over $2^O$ is \df{monotonic} if $S' \succeq_i S$ for any $S \subseteq S' \subseteq \mathcal{O}$. Throughout the paper, we focus on allocation problems $(\mathcal{N}, \mathcal{O}, (\succeq_i)_{i \in \mathcal{N}})$ in which each preference $\succeq_i$ is monotonic. So we restrict attention to items that are ``goods,'' not ``bads.''

A second property that we shall need is that preferences over bundles have to be responsive to preferences over individual items. We say that a preference  $\succeq$ over $2^{\mathcal{O}}$ is \textbf{responsive} if, for any $C, D \subseteq \mathcal{O}$ with $|C|=|D|$,
$$\exists \text{ bijection } f: C \to D \text{ s.t. } \{o\} \succeq \{f(o)\} \text{ for all } o\in C \implies C \succeq D.$$
If an agent's preferences are responsive, then a ranking $R$ over individual objects is defined by 
$$o \mathrel{R} o' \iff \{o\} \succeq \{o'\}.$$
For any two bundles of the same size, the existence of an object-by-object comparison in which every object in one bundle is ranked higher than the corresponding object in the other implies that the bundle with the better objects is preferred. Responsive preferences allow us to compare some bundles of different sizes as well. If $|C| > |D|$, we search for a subset $C' \subsetneq C$ with $|C'| = |D|$ such that $C' \succeq D$. Then, $C \succeq C' \succeq D$ by the preference monotonicity. Responsive preferences are used in our Proposition~\ref{prop:draft_find_ef1x}.

A strong version of the responsiveness property requires that preferences over bundles be obtained as the sum of the utilities over individual items. This turns out to be equivalent to the strong separability property. 
We say that a preference $\succeq$ is \df{additively separable} if there exists a \df{utility function} $u:\Ob\to\Re$ such that $S\succeq S'$ if and only if \[
\sum_{o\in S} u(o)\geq \sum_{o\in S'} u(o).
\]

\paragraph{Fairness.}

An agent \df{envies} another agent in an allocation if they prefer the other agent's bundle to their own. In consequence, we say that an allocation $\mu$ is \textbf{envy-free} if $\forall i , j \in \mathcal{N}$, $\mu(i) \succeq_i \mu(j)$. Envy-freeness is a popular notion of fairness in economics, but it is too restrictive for discrete allocation problems. It is often impossible for an allocation to be envy-free when multiple agents prefer the same objects. Thus, the discrete allocation literature has studied relaxations of envy-freeness.

To relax envy-freeness is to allow for the presence of envy, but to somehow limit, or bound, the envy that agents feel over the assignment of other agents. In a model in which preferences admit an objective measure of utility, it is possible to quantitatively measure and compare the envy that agents feel. We consider problems where such a cardinal measure of envy is unavailable. The bound must therefore be obtained from an objective manipulation of the bundles of objects in question.

The most widely studied relaxation is \textbf{envy-freeness up to one good (EF1)}. An allocation $\mu$ is EF1 if $\forall i, j \in \mathcal{N}$, 
$$
\mu(j) \succ_i \mu(i) \implies \exists o \in \mu(j) 
\text{ such that } \mu(i) \succeq_i \mu(j) \backslash\{o\}.
$$
Observe how EF1 imposes an upper bound on envy: agent $i$ may envy agent $j$, but the envy is limited in that removing just one object from $j$'s bundle is enough to eliminate it.

For some applications, however, dropping an object from an agent's bundle may be unthinkable, so the EF1 can potentially be non-binding. For example, consider an allocation of similar objects where fairness requires distributing an equal number of objects to each agent. If dropping an object from an agent's bundle represents replacing it with an outside option which is highly undesirable, then it is difficult to argue that the envy is moderate simply because its magnitude is bounded by the value of a single object.

Accordingly, an alternative relaxation of envy-freeness that is independent of the outside option value has been proposed by \cite{bogomolnaia2025teams} and \cite{caragiannis2024repeatedly}. The idea is that an agent's envy toward another agent is bounded by a single exchange of their objects:
\begin{definition}
\label{def:swapEF}
An allocation $\mu$ is \textbf{swap envy-free (swapEF)} if $\forall i, j \in \mathcal{N}$, 
$$
\mu(j) \succ_i \mu(i) \implies \exists o \in \mu(j), o'\in \mu(i)
\text{ such that } \mu(i) \cup \{o\} \backslash\{o'\} \succeq_i \mu(j) \cup \{o'\} \backslash\{o\}.
$$
\end{definition}

Perhaps swapEF is a more suitable requirement than EF1 when the objects are similar to one another relative to an outside option. Consider an allocation $\mu$ of houses, cars, and boats to two agents $i$ and $j$. Suppose agent $i$ gets a small house, a small car, and a small boat ($\mu(i) = (h, c, b)$), whereas agent $j$ gets a large house, a large car, and a large boat ($\mu(j) = (H, C, B)$). The allocation appears unfair but it can still satisfy EF1. If both agents have the same additively separable utilities, 10 for a large object and 9 for a small one, then agent $i$’s utility is $9 + 9 + 9 = 27$, while agent $j$’s utility is $10 + 10 + 10 = 30$. The allocation is EF1 because agent $i$ no longer envies agent $j$ if any object is removed from $\mu(j)$: $9 + 9 + 9 > 10 + 10$. Nonetheless, the allocation is not swapEF, because exchanging houses, cars, or boats would not eliminate the envy: $9 + 9 + 10 < 10 + 10 + 9$. 

Logically, EF1 and swapEF are unrelated concepts. One does not imply the other. If the agents' additively separable utilities in the house-car-boat example are 4 for a large house, 2 for a large car or boat, and 1 for a small object, then the allocation $\mu$ is swapEF but not EF1. The allocation is swapEF because agent $i$'s envy is eliminated by an exchange of houses: $u(H, c, b) = 6 > u(h, C, B) = 5$. It is not EF1, because agent $i$ still prefers agent $j$'s bundle even after dropping any single object: $u(h, c, b) = 3 < 4 \leq \min\{u(H,C), u(H,B), u(C,B)\}$. Hence, we may require a fair allocation to satisfy both EF1 and swapEF.

The main conceptual innovation in our paper is a notion of bounded envy that we term swap bounded envy. We defer its definition to Section~\ref{sec:multi-dimension}.

\section{Fair allocation mechanisms}
\label{sec:draftmech}

\subsection{Equal-number environment}
A swapEF allocation may not exist if commonly preferred objects cannot be evenly distributed. For example, if there are two agents but only one object, one agent must envy the other in any allocation, and there is no feasible exchange, so a swapEF allocation does not exist. This non-existence problem persists when there are three objects. Suppose the agents have identical additively separable utilities: $u_o=3$, $u_{o'}=4$, $u_{o''}=5$. Then, the utility from the bundle of two objects is $3 + 4 = 7$ or higher, so the agent who gets one object envies the other agent who gets two, and the envy persists after any exchange of objects.

Accordingly, we focus on \textbf{an equal-number allocation environment}. The number of objects $m$ and the number of agents $n$ satisfy $m=n K$ for some $K \in \mathbb{N}$ so that the objects can be evenly distributed. This setting includes multi-dimensional allocations such as house-car-boat allocations or multi-year campus housing allocations (the application we emphasized in the introduction).\footnote{While expositionally more complicated, the same results can be obtained when agents have $K$-demand preferences, meaning they prefer to receive up to $K$ objects.}

A \textbf{draft mechanism} takes agents' preferences and allocates objects over $K$ rounds. In each round, agents are ordered by a priority, and following this ordering, they sequentially select their most preferred object from the remaining ones. The priority ordering can be arbitrary, period-specific, and may even depend on the preference profile. Examples are:

\begin{itemize}
    \item Serial-dictatorship (SD): Fix an ordering of the agents. In each round, the agents pick an object given order.
    \item Alternating SD: Fix an ordering $>$ of the agents. In the first round, agents pick an object in the order given by $>$. In the next round, they pick objects in the opposite order to $>$. Then they alternate orders in each round.
    \item First-in-first-out (FIFO): Fix an ordering $>_1$ of the agents. Given an ordering $>_k$ used in round $k$, let $>_{k+1}$ be obtained by moving the first agent in $>_k$ to the last position, keeping the rest the same.
\end{itemize}
There are random versions of each of the preceding mechanisms obtained by choosing the initial ordering at random.

\begin{proposition}
\label{prop:draft_find_ef1x}
Suppose $m = nK$ for some $K \in \mathbb{N}$, and that every agent $i$ has a responsive preference $\succeq_i$ over $2^{\mathcal{O}}$. Then, any draft mechanism produces an allocation that satisfies both EF1 and swapEF.
\end{proposition}

Proposition~\ref{prop:draft_find_ef1x} is stated in \cite{bogomolnaia2025teams} for the case of additively separable preferences (see their Proposition 6). Their proof, however, only requires preferences to be responsive. For completeness, we present a self-contained proof below.

\begin{proof}
Let $\mu$ be an allocation obtained by a draft mechanism. For any agents $i$ and $j$, we denote $\mu(i) = (o^1_i, o^2_i, \dots, o^K_i)$, where agent $i$ obtained object $o^k_i$ in round $k$, and similarly write $\mu(j) = (o^1_j, o^2_j, \dots, o^K_j)$ for agent $j$. 

No matter how agents are ordered in each round $k = 1, 2, \dots, K-1$, we have $\{o^k_i\} \succeq_i \{o^{k+1}_j\}$, because $o^k_i$ was chosen when $o^{k+1}_j$ was available.

Now, there are two possibilities:
\begin{enumerate}
    \item $\{o^K_i\} \succeq_i \{o^1_j\}$: In this case, it follows from responsive preferences that $\mu(i) \succeq_i \mu(j)$.

    \item $\{o^1_j\} \succ_i \{o^K_i\}$: In this case, agent $i$ may envy $j$. However, an exchange of $o^1_j$ and $o^K_i$ results in
    \begin{align*}
        \mu(i) \cup \{o^1_j\} \setminus \{o^K_i\} &= (o^1_i, o^2_i, \dots, o^{K-1}_i, o^1_j) \\
        &\succeq_i (o^2_j, o^3_j, \dots, o^K_j, o^K_i) = \mu(j) \cup \{o^K_i\} \setminus \{o^1_j\},
    \end{align*}
    So, the allocation is swapEF.

    Moreover,
    \begin{align*}
        \mu(i) &= (o^1_i, o^2_i, \dots, o^{K-1}_i, o^K_i) \\
        &\succeq_i (o^2_j, o^3_j, \dots, o^K_j, o^K_i) = \mu(j) \cup \{o^K_i\} \setminus \{o^1_j\} \\
        &\succeq_i (o^2_j, o^3_j, \dots, o^K_j) = \mu(j) \setminus \{o^1_j\},
    \end{align*}
    where the last inequality follows from the preference monotonicity with respect to set inclusion. Thus, the allocation is EF1 as well.
\end{enumerate}
\end{proof}

Unfortunately, an allocation produced by a draft mechanism may not be Pareto efficient:
\begin{example}
Suppose there are $n = 2$ agents and $m = 6$ objects. Agents $i$ and $j$ have additively separable utilities: $u_i = (10, 9, 5, 4, 1, 0)$ and $u_j = (20, 4, 3, 2, 1, 0)$. Under the SD draft mechanism with a priority order $i > j$, agent $i$ receives a utility of 16 ($= 10 + 5 + 1$), while agent $j$ receives a utility of 6 ($= 4 + 2 + 0$). However, this allocation is Pareto dominated by another allocation that gives agent $i$ a utility of 18 ($= 9 + 5 + 4$) and agent $j$ a utility of 21 ($= 20 + 1 + 0$).
\end{example}

\subsection{Multi-dimensional environment}
\label{sec:multi-dimension}

A special case of the equal-number allocation environment is the multi-dimensional environment -- for example, the house-car-boat economy or college dorm allocation over multiple years as in the example with Alice and Bob that we discussed in the introduction. We refer to the dimensions as \emph{time periods}, denoted by $t = 1, 2, \dots, T$. 

We assume that there is a set $O_t$ of $n$ objects available in period $t$. The objects could be, for example, dorm rooms. So the same physical object is available in different periods, but we keep them apart as distinct because their consumption takes place at different periods. We assume also the existence of a ``null object,'' $\emptyset_t$, representing an outside option for agents who receive no object in period $t$. The null object is  available in unlimited supply. Let $\bar O_t=O_t\cup \{\emptyset_t\}$.

An agent consumes a tuple $(o_1, \dots, o_T)\in\prod_{t=1}^T\bar O_t$, where each $o_t$ is either an object available in period $t$ or the null object $\emptyset_t$. We may also write a tuple as $(o_t, o_{-t})$. We continue to assume responsive preferences. Thus, there is a strict preference $R_i$ over $\cup_t \bar O_t$ such that $(o_t, o_{-t}) \succeq (o'_t, o_{-t})$ iff $o_t \mathrel{R_i} o'_t$. Hence, if $o_t \mathrel{R_i} o'_t$ for all $t$, then $(o_1, \dots, o_T) \succeq_i (o'_1, \dots, o'_T)$. Abusing notation, we denote $R_i$ by $\succeq_i$. Note that $o_t R_i \emptyset_t$, as we assume each preference to be monotonic throughout the paper.

We also abuse notation in that we may omit consumption in one period. The implicit assumption in those cases is that the agent consumes the outside option. For example, when $T=3$ and we write consumption as $(o_1,o_2)$ we really mean $(o_1,o_2,\emptyset_3)$.

A multi-dimensional allocation is a function $\mu = (\mu_t)$ with $\mu_t:\mathcal{N} \to \bar O_t$ such that $\mu_t(j) \neq \mu_t(i)$ unless $\mu_t(j) = \mu_t(i)=\emptyset_t$ for all $j \neq i$. A multi-dimensional allocation $\mu$ is 
\begin{itemize}
    \item \textbf{envy-free up to one good (EF1)} if $\forall i, j$, 
\[
\mu(j) \succ_i \mu(i) \implies \exists t \text{ such that } \mu(i) \succeq_i (\os_t, \mu_{-t}(j)); \text{ and }
\] 
    \item \textbf{swap envy-free (swapEF)} if $\forall i, j$, 
\[
\mu(j) \succ_i \mu(i) \implies \exists t \text{ such that } (\mu_t(j), \mu_{-t}(i)) \succeq_i (\mu_t(i), \mu_{-t}(j)).
\]
\end{itemize}

\subsubsection{Draft mechanisms are not EF1 or swapEF in a multi-dimensional environment.}
Draft mechanisms can be suitably modified for a multi-dimensional environment by allocating objects in period~$t$ during round~$t$. 
However, unlike in an equal-number allocation environment (Proposition~\ref{prop:draft_find_ef1x}), draft mechanisms modified for a multi-dimensional setup are no longer EF1 or swapEF. For instance, consider the SD or FIFO mechanism, and suppose that there are fewer periods than agents ($T < n$). Then, there exist agents, say $i$ and $j$, such that agent~$j$ picks an object before agent~$i$ in every round $t = 1, \dots, T$. Agent~$i$ may envy agent~$j$, but there may be no way to eliminate the envy by dropping one object from $\mu(j)$ or by exchanging $\mu_t(i)$ and $\mu_t(j)$ in some period $t$.

Similarly, alternating SD can be EF1 and swapEF only in special cases.
\begin{lemma}\label{lem:SD}
\begin{enumerate}
    \item\label{it:lemSD1} If $T=2$, alternating SD is EF1 and swapEF.
    \item\label{it:lemSD2} If $T \leq 4$ and agents have additive separable utilities, then alternating SD is swapEF.
\end{enumerate}
\end{lemma}

\begin{proof}

\textbf{(\underline{Statement~\ref{it:lemSD1}}.)}

Suppose $i$ envies $j$: $\mu(j) \succ_i \mu(i)$. Since the ordering of agents is alternating, in either round $t = 1$ or $t = 2$, agent $i$ picks before agent $j$, so $\mu_t(i) \succeq_i \mu_t(j)$.

First, $\mu_{t'}(i) \succeq_i \emptyset_{t'}$ for $t' \neq t$, because preferences are monotonic in set inclusion. Then, by responsiveness, $(\mu_t(i), \mu_{t'}(i)) \succeq_i (\mu_t(j), \emptyset_{t'})$, i.e., dropping an object from $\mu(j)$ eliminates agent $i$'s envy, so the allocation $\mu$ is EF1.

Second, since $i$ envies $j$ despite $\mu_t(i) \succeq_i \mu_t(j)$, it must be that $\mu_{t'}(j) \succ_i \mu_{t'}(i)$ for $t' \neq t$. Thus, $(\mu_t(i), \mu_{t'}(j)) \succeq_i (\mu_t(j), \mu_{t'}(i))$, i.e., an exchange in period $t'$ eliminates agent $i$'s envy, so the allocation $\mu$ is swapEF.

\textbf{(\underline{Statement~\ref{it:lemSD2}}.)}

Suppose $i$ envies $j$: $\sum_t u_i(\mu_t(i)) < \sum_t u_i(\mu_t(j))$. Since $T \leq 4$, agent $j$ picks objects before agent $i$ in up to two periods, say $t$ and $t'$. For any other period $t'' \neq t, t'$, agent $i$ prefers $\mu_{t''}(i)$ over $\mu_{t''}(j)$.

Without loss of generality, suppose that $u_i(\mu_t(j)) - u_i(\mu_t(i)) \geq u_i(\mu_{t'}(j)) - u_i(\mu_{t'}(i))$. We assume the latter utility difference or both differences to be zero if $T \leq 3$, so that agent $j$ picks before $i$ in at most one period. The inequality means that agent $i$ prefers an exchange of period~$t$ objects over an exchange of period~$t'$ objects. Then,
$$
u_i(\mu_t(j)) + u_i(\mu_{t'}(i)) + \sum_{t'' \neq t, t'} u_i(\mu_{t''}(i)) \geq u_i(\mu_t(i)) + u_i(\mu_{t'}(j)) + \sum_{t'' \neq t, t'} u_i(\mu_{t''}(j)).
$$
Thus, an exchange in period~$t$ eliminates agent $i$'s envy.
\end{proof}

In general, alternating SD mechanism is neither EF1 nor swapEF.

\begin{example}[If $T \geq 3$, alternating SD is not EF1]
Suppose there are three agents, ordered $i > j > k$, and three objects $O_t = \{o, o', o''\}$ in each period $t = 1, 2, 3$. Agents have identical additively separable utilities $\sum_{t=1}^{3} u(o_t)$, where $u(o) = 100$, $u(o') = 2$, and $u(o'') = 1$. We omit $t > 3$ by assuming zero utility for those periods. 

By alternating SD, agent $i$ obtains a utility of $100 + 1 + 100$, whereas agent $j$ receives a utility of $2 + 2 + 2$. Agent $j$ envies agent $i$, even after dropping one object from agent $i$'s bundle.
\end{example}

\begin{example}[If $T \geq 5$, alternating SD is not swapEF.]
\label{ex:alt_SD_violateEF1X}
Take the same ordering of three agents ($i > j > k$) and the object set $O_t = \{o, o', o''\}$ for $t = 1, 3, 5$. Assume agents have identical additively separable utilities $\sum_{t = 1,3,5} u(o_t)$, with $u(o) = 100$, $u(o') = 2$, $u(o'') = 1$. We omit other periods $t \neq 1, 3, 5$ by assuming zero utility. 

By alternating SD, agent $i$ receives a utility of $100 + 100 + 100$, whereas agent $j$ receives $2 + 2 + 2$. Agent $j$ envies agent $i$, even after an exchange in any period.
\end{example}

\subsubsection{A TTC+SD mechanism}\label{sec:TTCSD}

We may now formally define swap bounded envy, a concept that we discussed in the introduction. A draft mechanism may not produce an EF1 and swapEF allocation in a multi-dimensional environment. We show that swapBE is obtained through a modification of SD based on a reallocation of bundles in each round of SD. In fact, the bound on envy is obtained round by round.

\begin{definition}
    \label{def:swapBE}
    A multi-dimensional allocation $\mu$ has \textbf{swap bounded envy (swapBE)} if $\forall i,j \in \mathcal{N}$, 
\begin{align*}
\mu(j) \succ_i \mu(i) \implies & \exists (o_1, \ldots, o_T) \in \Pi_t O_t \text{ and } t\in\{1, \dots, T\} \text{ such that } \\
& \mu(i) \succeq_i (o_1, \ldots, o_T) \text{ and } (\mu_t(j), o_{-t}) \succeq_i (o_t, \mu_{-t}(j)).
\end{align*}
\end{definition}

The motivation behind swapBE is similar to the motivation for swapEF, to bound envy by a one-period exchange. The difference is that, while swapEF considers a direct exchange between $\mu(i)$ and $\mu(j)$, swapBE identifies a reference bundle $(o_1, \dots, o_T)$ that is ranked below $\mu(i)$ in agent $i$'s preference list, but still lies within a one-exchange distance from $\mu(j)$.

\begin{remark}
\label{rmk:TTC_swapEF}
SwapBE is, arguably, more appealing than swapEF because it is robust to a reallocation that leaves the envying agent better off. To illustrate this, suppose agents $i$ and $j$ receive objects over 4 periods, and agent $i$'s additively separable utilities are $u_i(\mu(j)) = 80 + 80 + 80 + 45 > u_i(\mu(i)) = 100 + 97 + 1 + 1$. Thus, agent $i$ envies agent $j$, and a single exchange can eliminate this envy. Now, suppose a reallocation improves agent $i$'s bundle to $u_i(\mu'(i)) = 50 + 50 + 50 + 50$, while agent $j$'s bundle remains unchanged. swapEF no longer holds, because agent $i$ still envies $j$, but the envy can no longer be eliminated through a single exchange. Still, agent $i$'s envy toward $j$ is bounded by one exchange.

When a reallocation improves agent $i$'s bundle while agent $j$'s remains unchanged, both EF1 and swapBE are preserved, but swapEF may be violated.
\end{remark}

SwapBE is essentially a weaker requirement than swapEF. If a multi-dimensional allocation $\mu$ is swapEF, then $\mu(i)$ can be used as the bundle $(o_1, \dots, o_T)$ in an exchange with $\mu(j)$ in Definition~\ref{def:swapBE} (as long as $i$ does not consume the null object in any period in $\mu$).

\begin{remark}
Observe that the  reference bundle $(o_1, \dots, o_T)$ in Definition~\ref{def:swapBE} may not include an outside option in any period. Otherwise, an EF1 allocation would trivially satisfy swapBE. If an agent $i$ envies agent $j$ (i.e., $\mu(j) \succ_i \mu(i)$), and dropping an object $\mu_t(j)$ eliminates the envy (i.e., $\mu(i) \succeq_i (\emptyset_t, \mu_{-t}(j))$), then the bundle $(\emptyset_t, \mu_{-t}(j))$ is less preferred than $\mu(i)$, yet it can still eliminate the envy after a swap between $\emptyset_t$ and $\mu_t(j)$.
\end{remark}

\begin{remark}
    The properties of being swapBE and EF1 are logically unrelated. A swapBE allocation may not be EF1: See the example after Definition~\ref{def:swapEF}, where a swapEF allocation (which satisfies swapBE) violates EF1. And an EF1 allocation may not be swapBE: For an example, consider two agents and two houses allocated over three periods. Agent $1$ values house $1$ at $4$ and house $2$ at $5$ in each period. Let $\mu(1)=(1, 1, 1)$ and $\mu(2)=(2, 2, 2)$. Agent $1$ envies $2$. Removing house $2$ in any period from agent $2$'s bundle eliminates this envy, so $\mu$ is EF1. However, there is no bundle that is weakly less preferred (without an outside option) and an exchange that eliminates envy. Therefore, $\mu$ is not swapBE.
\end{remark}

We propose a mechanism that produces a multi-dimensional allocation that is EF1 and satisfies swapBE. A standard draft mechanism, which is not EF1, does not satisfy swapBE either. For example, Alternating SD in Example~\ref{ex:alt_SD_violateEF1X} does not satisfy swapBE. Agent~$i$ envies agent~$j$, because $u_i(\mu(i)) = 2 + 2 + 2$, whereas $u_i(\mu(j)) = 100 + 100 + 100$, and any bundle with utility less than or equal to $6$ cannot be within one-exchange distance from $\mu(j)$.

The idea behind the algorithm we propose for finding an EF1 and swapBE allocation is similar to that of the envy-graph procedure of \cite{lipton2004approximately}. When there are no multi-dimensional or equal-number constraints, the envy-graph procedure allocates objects to agents one by one, giving each object to an agent who is not envied by anyone else. If there is no such agent, the envy graph must contain a cycle, which can be eliminated by trading bundles among the agents in the cycle. Iteratively removing cycles ultimately identifies a non-envied agent who can then receive an object. If agents' preferences are separable, this procedure maintains the EF1 property of the allocation at every round.

In a multi-dimensional environment, objects are allocated over $T$ rounds. In each round $t$, we take the bundles allocated before round $t$ and apply the Top Trading Cycles (TTC) algorithm to eliminate envy cycles. Subsequently, the new objects in $O_t$ are allocated using Serial Dictatorship (SD). Non-envied agents -- those who received bundles in the last step of TTC -- pick first from $O_t$. Then, agents who are non-envied by the remaining agents -- those who received bundles in the second-to-last step of TTC -- pick from the remaining objects in $O_t$.\footnote{While we use TTC, any procedure that eliminates envy cycles would suffice. Once the envy cycles are removed, we can identify non-envied agents and, iteratively, identify agents who are non-envied by remaining others after removing the previously identified non-envied agents.}

Formally, \textbf{the TTC+SD mechanism} starts with an arbitrary allocation $\mu_1$ of period-1 objects. Since only period-1 objects are allocated, the initial allocation $\mu_1$ is trivially EF1 and has swapBE. 

Each round $t =2, \dots, T$ starts with an allocation $\mu_{<t}$ of objects from periods $1, \dots, t-1$. The procedure then runs TTC by constructing a graph in which the agents $\mathcal{N}$ and the bundles they have received, $\{ \mu_{<t}(i) \mid i \in \mathcal{N} \}$, are the vertices. There is an edge from $\mu_{<t}(i)$ to $i$ for each $i$, and an edge from $i$ to the agent's most preferred bundle, say $\mu_{<t}(j)$. Each step of TTC clears all cycles by exchanging the bundles.

Observe that if the input allocation $\mu_{<t}$ is EF1 and has swapBE, then the output allocation $\mu^{\text{ttc}}_{<t}$ also remains EF1 and has swapBE. Since TTC reallocates bundles rather than individual objects, agent $i$'s envy toward $j$ may simply shift to envy toward another agent. Moreover, agent $i$ gets better by TTC, so the envy can still be eliminated by dropping an object, or bounded by one exchange.

Next, we allocate period~$t$ objects using a SD procedure with an order taken from the outcome of the TTC. The previous TTC step yields a partition of agents into groups $A_1, A_2, \dots$, where $A_k$ is the $k$-th cycle cleared in the application of TTC. We order agents so that those in $A_k$ precede those in $A_{k'}$ for $k' < k$, with an arbitrary order within each group. Then we let agents who obtained bundles in the last step of TTC pick period~$t$ objects first. Then, we set those agents and the period~$t$ objects they take aside, and let the agents who obtained bundles in the second-last step of TTC pick among the remaining objects in $O_t$, and so on. 
This SD produces $\mu_t$; so we obtain $\mu_{\leq t} = (\mu^{ttc}_{<t}, \mu_t)$.

If agent $i$ envied agent $j$ before in $\mu^{ttc}_{<t}$, then agent $i$ picks a period~$t$ object before agent $j$. So, the envy can still be removed by one drop, and bounded by one exchange. A new envy may arise because an agent $i$ picks a period~$t$ object after agent $j$, but this envy can be removed by dropping the period~$t$ object, or bounded by a period~$t$ exchange. 

We need a restriction on preferences to guarantee that past choices do not get overturned by later allocations. Say that a preference $\succeq_i$ is \df{directionally separable} if, for any $t\geq 2$, 
$o_t \in O_t$, and 
$S_{<t},S'_{<t} \in\prod_{\tau=1}^{t-1} O_\tau$,
$$
S_{<t} \succeq_i S'_{<t} \iff 
(S_{<t}, {o_t}) \succeq_i (S'_{<t}, o_t) .$$

\begin{proposition}
\label{prop:TTC_SD}
In a multi-dimensional environment, suppose that agents have responsive and directionally separable preferences. Then the TTC+SD mechanism is EF1 and has swapBE.
\end{proposition}

The TTC+SD mechanism may not satisfy swapEF because the TTC reallocation of bundles may not preserve swapEF (see Remark~\ref{rmk:TTC_swapEF}).

\begin{proof} 
\begin{claim}
If $\mu_{<t}$ is EF1 (swapBE), then $\mu^{ttc}_{<t}$ is EF1 (swapBE).
\end{claim}

To prove the claim, suppose that $\mu^{\text{ttc}}_{<t}(j) \succ_i \mu^{\text{ttc}}_{<t}(i)$, i.e., agent $i$ envies agent $j$. Since TTC is individually rational, and reallocates bundles rather than individual objects, for some agent $k \neq i$,
\[
\mu_{<t}(k) = \mu^{\text{ttc}}_{<t}(j) \succ_i \mu^{\text{ttc}}_{<t}(i) \succeq_i \mu_{<t}(i).
\]

If $\mu_{<t}$ is EF1, there exists a period $s \in [t-1]:=\{1, 2, \dots, t-1\}$ such that
\[
\mu^{\text{ttc}}_{<t}(i) \succeq_i \mu_{<t}(i) \succeq_i (\emptyset_s, \mu_{[t-1] \setminus s}(k)) = (\emptyset_s, \mu^{\text{ttc}}_{[t-1] \setminus s}(j)),
\]
which implies that $\mu^{\text{ttc}}_{<t}$ is EF1.

Similarly, if $\mu_{<t}$ has swapBE, there exists a bundle $(o_1, \ldots, o_{t-1}) \preceq_i \mu_{<t}(i) \preceq_i \mu^{\text{ttc}}_{<t}(i)$, less preferred but still making agent $i$ envy-free after an exchange with $\mu_{<t}(k)$ in some period $s < t$. Then,
\[
(o_s, \mu^{\text{ttc}}_{[t-1] \setminus s}(j)) = 
(o_s, \mu_{[t-1] \setminus s}(k))
\preceq_i (\mu_s(k), o_{[t-1] \setminus s}) =(\mu^{ttc}_s(j), o_{[t-1] \setminus s}),
\]
which implies that $\mu^{\text{ttc}}_{<t}$ is swapBE.

\begin{claim}
If $\mu^{ttc}_{<t}$ is EF1 (swapBE), then $\mu_{\leq t} = (\mu^{ttc}_{<t}, \mu_t)$ is EF1 (swapBE).
\end{claim}

\textbf{(EF1)} Suppose $\mu_{\leq t}(j) \succ_i \mu_{ \leq t}(i)$: agent $i$ envies $j$. If $\mu^{ttc}_{< t}(i) \succeq_i \mu^{ttc}_{< t}(j)$, then
\[\mu_{\leq t}(i) =(\mu^{ttc}_{< t}(i), \mu_t(i)) \succeq_i (\mu^{ttc}_{< t}(i), \emptyset_t) \succeq_i (\mu^{ttc}_{< t}(j), \emptyset_t).
\]
On the other hand, if $\mu^{ttc}_{< t}(j) \succ_i \mu^{ttc}_{< t}(i)$, then agent $i$ picked a period~$t$ object before $j$, so $\mu_t (i) \succeq_i \mu_t(j)$. Since $\mu^{ttc}_{<t}$ is EF1, there exists $s <t$ such that $\mu^{ttc}_{< t}(i) \succeq_i (\emptyset_s, \mu^{ttc}_{[t-1] \setminus s}(j))$. Thus, 
\begin{align*}
    \mu_{\leq t} (i) 
    & = (\mu^{ttc}_{< t} (i), \mu_t (i)) \succeq_i (\emptyset_s, \mu^{ttc}_{[t-1] \setminus s}(j), \mu_t(i)) \quad \text{(by directional separability)}\\
    & \succeq_i (\emptyset_s, \mu^{ttc}_{[t-1] \setminus s}(j), \mu_t(j)) = (\emptyset_s, \mu_{[t] \setminus s}(j)). \quad \text{(by responsiveness)}
\end{align*}
In either case, dropping one object eliminates $i$'s envy toward $j$, so the allocation $\mu_{\leq t}$ is EF1.

\textbf{(SwapBE)} Suppose $\mu_{\leq t}(j) \succ_i \mu_{ \leq t}(i)$: agent $i$ envies $j$. 

If $\mu^{ttc}_{< t}(i) \succeq_i \mu^{ttc}_{< t}(j)$, then by directional separability and responsiveness, the envy must be from $\mu_t(j) \succ_i \mu_t(i)$. Then, $(\mu^{ttc}_{< t}(i), \mu_t(j)) \succeq_i (\mu^{ttc}_{< t}(j), \mu_t(i))$, so the envy is eliminated by one exchange.

On the other hand, if $\mu^{ttc}_{< t}(j) \succ_i \mu^{ttc}_{< t}(i)$, then agent $i$ picks a period~$t$ object before $j$, so $\mu_t (i) \succeq_i \mu_t(j)$. Since $\mu^{ttc}_{<t}$ is swapBE, there exists a bundle $(o_1, \dots, o_{t-1}) \preceq_i \mu^{ttc}_{<t}(i)$, less preferred but making agent $i$ envy-free after an exchange with $ \mu^{ttc}_{<t}(j)$ in some period $s < t$:
\[
(\mu^{ttc}_s(j), o_{[t-1] \setminus s}) \succeq_i (o_s, \mu^{ttc}_{[t-1] \setminus s}(j)).
\]
Then, by responsiveness and directional separability,
\[
(o_1, \dots, o_{t-1}, \mu_t(j)) \preceq_i (o_1, \dots, o_{t-1}, \mu_t(i)) \preceq_i (\mu^{ttc}_{<t}(i), \mu_t(i)) = \mu_{ \leq t}(i)
\]
and, but directional separability, 
\[
(\mu^{ttc}_s(j), o_{[t-1] \setminus s}, \mu_t(j)) \succeq_i (o_s, \mu^{ttc}_{[t-1] \setminus s}(j), \mu_t(j)).
\]
The bundle $(o_1, \dots, o_{t-1}, \mu_t(j))$ is less preferred to $\mu_{\leq t}(i)$ but can make agent $i$ envy-free after an exchange with $\mu_{\leq t} (j)$ in period~$s$. Thus, $\mu_{\leq t}$ has swapBE.
\end{proof}

\section{Social welfare maximization}
\label{sec:nashwelfare}
Neither the draft mechanisms in equal-number allocations, or the TTC+SD mechanism in multi-dimensional allocations, is designed to obtain a Pareto efficient allocation. Due to the Pareto inefficiency of those mechanisms, we resort to an alternative approach to obtain an allocation that is both Pareto efficient, and fair in the sense of EF1, swapEF, or swapBE.

In a general allocation problem that allows for an uneven distribution of objects, when agents have additive separable utilities, Nash-welfare maximization is known to obtain an EF1 allocation that is also Pareto efficient. To elaborate, let each agent's preference $\succeq_i$ be represented by additively separable utilities $u_i: \mathcal{O} \to \mathbf{R}_+$, and the payoff from an allocation $\mu$ be $\sum_{o \in \mu(i)} u_i(o)$. The \textbf{Nash welfare} of $\mu$ is 
$$NW(\mu)= \prod_i \left(\sum_{o \in \mu(i)} u_i(o) \right).$$ 
A maximum Nash-welfare allocation is EF1 and Pareto efficient (\cite{unreasonable_fairness}).\footnote{Pareto efficiency of a maximum Nash-welfare allocation is easy to verify. To understand the proof of EF1, suppose an allocation gives one object to agent 1 with utility $a > 0$, and two objects to agent 2 with utilities $x + y > 0$, yielding Nash welfare $a(x + y)$. Let $a', x', y'$ be the utilities if the corresponding object is allocated differently—either to agent 2 instead of agent 1, or vice versa. Suppose the allocation is not EF1 because agent 1 envies agent 2 ($a < x' + y'$), even after removing either object from agent 2’s bundle ($a < x'$ and $a < y'$). If $\frac{x}{x'} \leq \frac{y}{y'}$ (we omit the other case), then $a x < y' x \leq y x' \implies a(x + y) < (a + x') y$, so the allocation is not Nash-welfare maximizing.}

\subsection{Equal-number environment}

We observed before that an uneven distribution of objects can violate swapEF, which is why we have focused on an equal-number allocation environment. For the same reason, if Pareto efficient allocations (including maximum Nash-welfare allocations) require an unequal distribution of objects, no efficient allocation would satisfy swapEF:
\begin{example}
There are $n = 3$ agents and $m = 6$ objects. 
Agents have additively separable preferences, where the utility from a liked object is $1$, and $0$ otherwise. Agent $i$ likes only $o_1$ and $o_3$, agent $j$ likes only $o_2$ and $o_3$, while agent $k$ likes only $o_4$, $o_5$, and $o_6$. In a maximum Nash welfare allocation, and indeed in any Pareto efficient allocation, agent $k$ receives $o_4$, $o_5$, and $o_6$, agent $i$ receives $o_1$, and agent $j$ receives $o_2$, with $o_3$ allocated to either agent $i$ or $j$. In either case, agent $i$ (or $j$) envies agent $j$ (or $i$), even after an exchange of one good.
\end{example}

Hence, we restrict attention to the set of equal-number allocations, where each agent receives exactly $K$ objects:
$$
\mathcal{E} = \{ \mu : |\mu(i)| = K \text{ for all } i \in \mathcal{N} \}
$$
Then, for the case of 0-1 utilities, maximizing Nash welfare over $\mathcal{E}$ yields a Pareto-efficient allocation within $\mathcal{E}$ that is EF1 and swapEF. 

The following result was obtained by \cite{bogomolnaia2025teams}. See their Proposition 15. Again we include a proof for completeness.

\begin{proposition}\label{prop:maxNashW}
Suppose $m = nK$ for some $K \in \mathbb{N}$, and that each preference $\succeq_i$ over $2^{\mathcal{O}}$ can be represented by additive utilities taking values in $\{0,1\}$. If an allocation $\mu$ maximizes Nash welfare in $\mathcal{E} = \{ \mu : |\mu(i)| = K \text{ for all } i \in \mathcal{N} \}$ and $NW(\mu) > 0$,

then $\mu$ is EF1, swapEF, and Pareto efficient within $\mathcal{E}$.
\end{proposition}

\begin{proof}
Suppose that an allocation $\mu$ obtains the maximum Nash welfare in $\mathcal{E}$ and that $NW(\mu) > 0$. It is easy to verify that $\mu$ is Pareto efficient in $\mathcal{E}$, so we omit the proof.

To prove that $\mu$ is EF1 and swapEF, suppose that agent $i$ envies agent $j$, i.e., $\mu(j) \succ_i \mu(i)$. Then:

\begin{enumerate}
\item $u_i(\mu(j)) \geq u_i(\mu(i))+1$, because $i$ envies $j$. Moreover, since both agents receive exactly $K$ objects, there must be an object $o_i \in \mu(i)$ such that $u_i(o_i) = 0$.
\item $u_j(\mu(j)) \geq u_i(\mu(j))$. Otherwise, there exists an object $o' \in \mu(j)$ such that $u_j(o') = 0$ but $u_i(o') = 1$. Then, by exchanging $o_i \in \mu(i)$ and $o' \in \mu(j)$, we would obtain strictly higher Nash welfare than $NW(\mu)$.
\item In fact, $u_i(\mu(j)) = u_i(\mu(i)) + 1$. Otherwise, if $u_i(\mu(j)) > u_i(\mu(i)) + 1$, then exchanging $o_i \in \mu(i)$ with any $o' \in \mu(j)$ such that $u_i(o') = 1$ yields:
\begin{align*}
u_i(\mu(i)) \times u_j(\mu(j)) & < (u_i(\mu(i)) + 1) \times (u_j(\mu(j)) - 1) \\
& \leq u_i(\mu(i) \cup \{o'\} \setminus \{o_i\}) \times u_j(\mu(j) \cup \{o_i\} \setminus \{o'\}),
\end{align*}
contradicting the Nash-welfare maximization of $\mu$.
\end{enumerate}

Finally, take any object $o' \in \mu(j)$ such that $u_i(o') = 1$. Removing $o'$ from $\mu(j)$ results in $u_i(\mu(j) \setminus \{o'\}) = u_i(\mu(i))$, eliminating agent $i$'s envy toward agent $j$. Hence, $\mu$ is EF1.

Moreover, exchanging $o_i \in \mu(i)$ and $o' \in \mu(j)$ results in $u_i(\mu(i) \cup \{o'\} \setminus \{o_i\}) = u_i(\mu(i)) + 1 = u_i(\mu(j)) > u_i(\mu(j) \cup \{o_i\} \setminus \{o'\})$, which also eliminates agent $i$'s envy toward agent $j$. Hence, $\mu$ is swapEF.
\end{proof}

\subsection{Multi-dimensional environment}

\subsubsection{Nash welfare maximization}
Nash welfare maximization does not satisfy swapBE, and therefore is not swapEF either, in multi-dimensional allocations:
\begin{example}
Suppose that there are three agents with additively separable utilities, and the same set of three objects $\{o, o', o''\}$ across six periods. Assume that $u_1(o)=1$, $u_1(o')=2$, and $u_1(o'') = 30$ for agent $1$, and $u_i(o)=1$, $u_i(o')=2$, and $u_i(o'') = 3000$ for agents $i=2, 3$. 

A maximum Nash welfare allocation gives the object $o''$ to agent 2 or 3, to each agent for three periods. The allocation does not have swapBE, because any bundle with utility $2 \times 6$ (getting $o'$ for six periods) or lower cannot be improved enough by a single exchange to match a bundle containing $o''$ for two or more periods.
\end{example}

An optimization tends to fail to produce a fair allocation, because if two agents prefer a few objects but only one has strong preferences, social welfare maximization tends to assign \emph{all} of the preferred objects to that agent.

However, even when agents have the same utility magnitudes, for example, when the utility from the most preferred object is the same for all agents, social welfare maximization still fails to be swapEF:

\begin{example}
\label{example:equal_utility_magnitudes}
Suppose that there are five agents $N=\{i, j, k_1, k_2, k_3\}$ and five objects in each of the three periods. Agents have additively separable utilities given as: 
$$
\begin{bmatrix}
        U_i \\
        U_j \\
        U_{k1} \\
        U_{k2}\\
        U_{k3}
\end{bmatrix}
=
\begin{bmatrix}
        2 & 3 & \underline{4} & 1 & 0\\
        \underline{2} & 4 & 3 & 1 & 0\\
        1 & \underline{4} & 3 & 2 & 0\\
        1 & 4 & 3 & \underline{2} & 0\\
        1 & 4 & 3 & 0 & \underline{2}\\
\end{bmatrix}
+
\begin{bmatrix}
        2 & 3 & \underline{4} & 1 & 0\\
        \underline{2} & 4 & 3 & 1 & 0\\        
        1 & 4 & 3 & \underline{2} & 0\\
        1 & \underline{4} & 3 & 2& 0\\
        1 & 4 & 3 & 0& \underline{2}\\
\end{bmatrix}
+
\begin{bmatrix}
        2 & 3 & \underline{4} & 1 & 0\\
        \underline{2} & 4 & 3 & 1 & 0\\
        1 & 4 & 3 & \underline{2}& 0\\
        1 & 4 & 3 & 0& \underline{2}\\
        1 & \underline{4} & 3 & 2 & 0\\
\end{bmatrix}
$$
For example, agent $i$'s utility from the first object in period 1 is 2, and the second object in period 2 is 3, and the fifth object in period 3 is 0. Although agents have heterogeneous preferences, the utilities from objects ranked the same are identical across agents. The allocation marked by underlines in the matrices maximizes Nash welfare, but it is not swapEF.
\end{example}

Verifying Nash welfare maximization in Example~\ref{example:equal_utility_magnitudes} needs a brute-force search, but the intuition is as follows: Agent $ i $'s most preferred object differs from those of the other agents, allowing her to obtain it in each period. The last two objects in each period (six objects in total from three periods) must be allocated to the three $ k $ agents, ensuring each receives a utility of $ 4 $ (i.e., $ 2 + 2 $). Thus, the most preferred object for the four competing agents in each period should go to a $ k $ agent rather than to agent $ j $, because without it, agent $ j $ can still secure a utility of $ 2 $ each period.

The allocation is not swapEF because:
\begin{itemize}
    \item $ j $ envies $ i $: $ 2 + 2 + 2 < 3 + 3 + 3 $,
    \item $ j $ still envies $ i $ after the period-1 exchange: $ 3 + 2 + 2 < 2 + 3 + 3 $,
    \item and similarly for any other one-period exchanges.
\end{itemize}

Similarly, a utilitarian-welfare maximizing allocation does not satisfy swapEF either. 

In Example~\ref{example:equal_utility_magnitudes}, the highest possible utilitarian welfare is $42$, but achieving this requires that, in each period, agent $i$ receives the third object (for a total utility of $4 \times 3$), agent $j$ receives the first object (for a total utility of $2 \times 3$), and each agent $k$ receives a utility of $4$, $2$, or $2$ (for a total utility of $(4 + 2 + 2) \times 3$). The utilitarian welfare maximizing allocation is not swapEF.

\subsubsection{Submodular welfare maximization}
\label{sec:submodular}

In a multi-dimensional environment, we show that when all agents have identical preferences, optimizing a submodular welfare function produces an allocation that is Pareto efficient, EF1, and swapEF (and therefore has swapBE).

Suppose that all agents have the same separable preferences $\succeq$ over bundles $\prod_t O_t$. Let $f: \prod_t O_t \to \Re$ be a social welfare function that is strictly increasing and submodular in $\succ$: if $o'_t\succ o_t$ and $o'_{-t}\succ o_{-t}$ then \[ 
f(o'_t,o'_{-t}) - f(o_t,o'_{-t}) < f(o'_t,o_{-t}) - f(o_t,o_{-t}).
\]
A social welfare of an allocation $\mu$ is $\sum_{i\in \mathcal{N}} f(\mu_1(i),\ldots, \mu_T(i))$. The submodularity assumption is kind of ``decreasing differences,'' meaning that the marginal benefit of improving the object in period $t$ by going from $o_t$ to $o'_t$ is greater when the consumption in the remaining periods is worse.

\begin{proposition}
If all agents have the same additively separable 
preferences over bundles, an allocation that maximizes a submodular social welfare is Pareto efficient, EF1, and swapEF (and therefore has swapBE).

\end{proposition}

\begin{proof}

\textbf{(Pareto efficiency)} The proof is trivial, so we omit. 

\textbf{(EF1)} Suppose that $\mu$ is Pareto efficient but not EF1. In particular, suppose that $\mu(j) \succ_i \mu(i)$ and that, for all $t$, $(\emptyset_t,\mu_{-t}(j)) \succ_i  \mu(i)$. Then, since $\mu_t(i) \succeq_i \emptyset_t$, we have $\mu_{-t}(j) \succ_i \mu_{-t}(i)$ for all $t$. 

There exists $t$ such that $\mu_t(j) \succ_i \mu_t(i)$, for otherwise $\mu_t(i) \succeq_i \mu_t(j)$ for all $t$, contradicting $\mu(j) \succ_i \mu(i)$. Now by submodularity we have:
\[ 
f(\mu_t(j),\mu_{-t}(i)) - f(\mu_t(i),\mu_{-t}(i) ) > f(\mu_t(j),\mu_{-t}(j)) - f(\mu_t(i),\mu_{-t}(j)).
\]

But then
\begin{align*}
& f(\mu(i)) + f(\mu(j)) +\sum_{h\neq i,j} f(\mu(h)) \\
& < f(\mu_t(i),\mu_{-t}(j)) + f(\mu_t(j),\mu_{-t}(i))+\sum_{h\neq i,j} f(\mu(h)),
\end{align*}
so $\mu$ is not social welfare maximizing.

\textbf{(SwapEF)} Suppose that $\mu$ is not swapEF. In particular, suppose that $\mu(j)\succ_i \mu(i)$ and that, for all $t$, $(\mu_t(i),\mu_{-t}(j)) \succ_i  (\mu_t(j),\mu_{-t}(i))$.

There exists $t$ such that $\mu_t(j) \succ_i \mu_t(i)$, because otherwise $\mu_t(i) \succeq_i \mu_t(j)$ for all $t$, which would contradict $\mu(j) \succ_i \mu(i)$. Then, due to additive separability, we must have $\mu_{-t}(j) \succ_i \mu_{-t}(i)$ to avoid contradicting $(\mu_t(i), \mu_{-t}(j)) \succ_i (\mu_t(j), \mu_{-t}(i))$.

Now by submodularity we have:
\[f(\mu_t(j),\mu_{-t}(i)) - f(\mu_t(i),\mu_{-t}(i) ) > f(\mu_t(j),\mu_{-t}(j)) - f(\mu_t(i),\mu_{-t}(j)).
\]

But then
\begin{align*}
& f(\mu(i)) + f(\mu(j)) +\sum_{h\neq i,j} f(\mu(h))\\
& <  f(\mu_t(i),\mu_{-t}(j)) + f(\mu_t(j),\mu_{-t}(i)) + \sum_{h\neq i,j} f(\mu(h)),
\end{align*}
so $\mu$ is not social welfare maximizing.
\end{proof}

Finally, we find an example of a social welfare function $f$ that satisfies the submodularity. Suppose that $(o_1, \dots, o_T) \succeq (o'_1, \dots, o'_T)$ iff $\sum_{t} b_t u(o_t)\geq \sum_{t} b_t u(o'_t)$ for some $b_t > 0$ for all $t$. For example, we could have time discount $b_t=\delta^t$ for some $\delta\in (0,1)$. Let
\[ 
f(u_1,\ldots, u_T) = A \sum_{t=1}^T b_t u_t - B\sum_{t=1}^T \al_t u_t \left( \sum_{\tau\neq t} \beta_\tau u_{\tau} \right),
\] 
where the parameters $A,B$, and $\al_t $ and $\beta_t$ are strictly positive, and  $A/B$ is large enough so that $f$ is monotone increasing.

The derivative with respect to $u_t$ is
\begin{align*}
D_{u_t}f(u_1,\ldots, u_T) & = A b_t  - B [ \al_t (\sum_{\tau\neq t} \beta_\tau u_{\tau}) + \sum_{\tau\neq t} \al_\tau u_\tau \beta_t ]\\
& = A b_t - B \sum_{\tau\neq t} (\al_t \beta_\tau  +  \al_\tau  \beta_t)u_\tau.
\end{align*}
Suppose that $b_t = \al_t=\beta_t$. Then, 
\[ 
D_{u_t}f(u_1,\ldots, u_T) =  A b_t - 2 B  b_t \sum_{\tau\neq t} b_\tau  u_\tau,
\] which is monotone decreasing in the magnitude $\sum_{\tau\neq t} b_\tau  u_\tau$.

\bibliographystyle{ecta}
\bibliography{alloc}

\end{document}